%% file: root.tex
\documentclass[journal,twoside,web]{ieeecolor}

\usepackage{generic}
\usepackage{cite}
\usepackage{amsmath,amssymb,amsfonts}
\usepackage{algorithmic}
\usepackage{graphicx}
\usepackage{textcomp}

\usepackage{graphicx}
\usepackage{subfigure}
\usepackage{epsfig} % for postscript graphics files
\usepackage{cite}
\usepackage{lcsys}

\newtheorem{theorem}{Theorem}[section]
\newtheorem{lemma}[theorem]{Lemma}
\newtheorem{proposition}[theorem]{Proposition}

\newtheorem{assumption}[theorem]{Assumption}
\input{preamble}

\begin{document}

\def\BibTeX{{\rm B\kern-.05em{\sc i\kern-.025em b}\kern-.08em
    T\kern-.1667em\lower.7ex\hbox{E}\kern-.125emX}}
\markboth{\journalname, VOL. XX, NO. XX, XXXX 2017}
{Author \MakeLowercase{\textit{et al.}}: Preparation of Papers for IEEE Control Systems Letters (September 2026)}

\title{A Frequency Domain Approach to Bounding Riccati Equation Perturbations}

\author{Rachel Newton, Laura Balzano, \IEEEmembership{Senior Member, IEEE}, and Peter Seiler, \IEEEmembership{Fellow, IEEE}
\thanks{This work is supported by ARO award W911NF-26-1-A219 (all), and NSF graduate research fellowship DGE 1841052 (RN only).}
\thanks{The authors are with the Electrical and Computer Engineering Department at the University of Michigan, Ann Arbor, MI 48109 (e-mails: rhoadesr@umich.edu, pseiler@umich.edu, girasole@umich.edu). }
}

\maketitle
\thispagestyle{empty}

%%%%%%%%%%%%%%%%%%%%%%%%%%%%%%%%%%%%%%%%%%%%%%%%%%%%%%%%%%%%%%%%%%%%%%%%%%%%%%%%
\begin{abstract}
The discrete algebraic Riccati equation (DARE) is used to solve for the optimal feedback gain in linear-quadratic regulator (LQR) control of discrete-time systems. In this letter, we consider the effect of perturbations to the LQR data on the DARE solutions. We present a new frequency domain approach to derive an explicit bound on the difference in DARE solutions under such perturbations. Our bounds are similar to other results that have recently appeared in the literature with two distinctions. First, we use a frequency domain constant that measures the degree of closed-loop stability, in contrast to other results that use a time domain measure. Second, we start from weaker detectability assumptions on the LQR state cost matrix. 

\end{abstract}


\section{Introduction}
\label{sec:introduction}

The linear-quadratic regulator (LQR) is a standard optimal control problem for which the optimal controller is a static state feedback gain. 
If the system and cost matrices are known, this gain can be computed from a solution to the discrete algebraic Riccati equation (DARE). 
However, in data-driven and robust control, the system data may not be exactly known. 

As a consequence, it is useful to derive bounds on the DARE solutions under perturbations to the state and cost matrices.

There are several DARE perturbation results in the literature. Konstantinov used an operator-theoretic approach to bound these solutions \cite{konstantinov1993perturbation}.
In \cite{mania2019certainty}, Mania et al. expanded on the operator-theoretic approach to construct a simplified bound, and derived a tighter bound under stronger assumptions on the LQR data. 
Finally, \cite{simchowitz2020naive} derived a bound using a self-bounding ODE method. Later results extend these types of bounds to certain classes of time-varying systems, e.g. \cite{sattar2022certainty, jadbabaie2024multi}.

In this letter, we also construct a computable bound on the DARE solutions when the LQR data is perturbed using a novel frequency domain approach. 
Our bound includes a frequency domain measure of closed-loop stability, $\chat_{A1}$ as defined in \eqref{eqn:cA1_def2}. 
In contrast, existing bounds use a time domain measure of stability $\tau$ as reviewed in Section \ref{sec:MainResults}. 
Furthermore, the bounds in the literature typically assume the state cost matrix $Q$ is either positive definite or satisfies an observability assumption, and so the DARE solution is positive definite. 
Our bound starts from the weaker assumptions that $Q$ is positive semidefinite and satisfies a detectability assumption. As a result, the DARE solutions can be positive semidefinite, allowing our bound to be applied in more general situations.

We first introduce the stochastic LQR problem as a framework for the DARE in Section \ref{sec:pbm_formulation}. We then construct a bound on the difference of the \htwonorms for two systems that are $\epsilon$-close in Section \ref{sec:h2normbounds}. 
Finally, in Section \ref{sec:MainResults}, we utilize the equivalence between the \htwonorm of a closed-loop system and the associated LQR cost to prove the DARE solutions must be close.
The result is an intuitive bound that only incorporates frequency domain properties as constants.

\section{Problem Formulation}\label{sec:pbm_formulation}

Consider a discrete-time, LTI state-space system
\begin{equation} \label{eqn:ltisys}
		x_{k+1} =Ax_k+B u_k + w_k, 
\end{equation}
where $x_k\in\R^{n_x}$ and $u_k\in\R^{n_u}$ are the state and input at time $k$. Let the process noise $w_k$ be i.i.d. zero-mean Gaussian with covariance matrix $\Sigma_w$. The initial condition $x_0$ is also assumed to be Gaussian with mean $\bar{x}_0$ and covariance $\Sigma_x$. The process noise $w_k$ and initial condition $x_0$ are independent.

Define the following infinite horizon, quadratic cost:
\begin{align}
\label{eqn:lqrCost}
J(u) := \lim_{N\to \infty} 
\mathbb{E} \left[ \frac{1}{N} \sum_{k=0}^N x_k^\top Q x_k + u_k^\top R u_k \right] .
\end{align}
Here $Q \succeq 0$ and $R \succ 0$ are cost matrices, selected by the designer, that specify the trade-off between state regulation and control effort. 
We can obtain the thin square root factor of $Q$ to give the form $Q=C^\top C$ where $C \in \R^{n_y \times n_x}$  with $n_y \le n_x$.
The infinite horizon LQR problem searches for the control signal $u_k$ that solves the following 
\begin{align}
\label{eqn:lqr}
\begin{split}
J^* := & \min_u \, J(u) \\
&\st \mbox{ Equation~\ref{eqn:ltisys}
}.
\end{split}
\end{align}
The control signal at time $k$ is assumed to depend only on the current and past states, i.e., $u_k$ only depends on $\{x_0,\ldots,x_k\}$. 
Throughout this letter, we will consider static state feedback controllers of the form $u_k = -K x_k$. We use $J(K)$ to denote the cost associated with gain $K$.
The optimal solution to \eqref{eqn:lqr} is given by one such static, state feedback gain. We first provide the assumptions required to formally state this result.

\vspace{0.1in}
\sloppypar\begin{assumption}\label{asmp:lqrReg}
    The discrete-time LTI system matrices $(A,B)$, noise covariance matrix $\Sigma_w$, and state cost matrix $Q=C^\top C$ and control cost matrix $R$ satisfy %(a) $Q=C^\top C \succeq 0$ and $R\succ 0$, (b) $\Sigma_w \succeq 0$, (c) $(A,B)$ stabilizable, and (d) $(A,C)$ detectable.

    \begin{enumerate}[label=(\alph*)]
    \item $Q=C^\top C \succeq 0$ and $R\succ 0$,
    \item $\Sigma_w \succeq 0$,
    \item $(A,B)$ stabilizable, and 
    \item $(C,A)$ detectable. 
    \end{enumerate}
\end{assumption}
\vspace{0.1in}

The optimal LQR solution is stated in the next lemma using these assumptions. 

\vspace{0.1in}
\begin{lemma}[Section 6.4.5 of \cite{kwakernaak1972linear}]
	\label{lem:LQR}    
    Consider discrete-time LTI system matrices $(A,B)$, noise covariance matrix $\Sigma_w$, and cost matrices $(Q,R)$ that satisfy Assumption~\ref{asmp:lqrReg}. Then there is a unique positive semidefinite solution $P^{\star}$  to the following discrete-time algebraic Riccati equation (DARE):
	\begin{equation}
		P-A^\top PA - Q + A^\top PB \, (R+B^\top PB)^{-1}\, B^\top PA  = 0.  \label{eqn:fulldare_lqr}
	\end{equation}
	Moreover, we can define the state feedback gain
	\begin{equation}
		K^{\star}:=\left(R+B^\top P^{\star} B\right)^{-1} B^\top P^\star A \label{eqn:Kf}
	\end{equation} 
	such that the controller $u_k = -K^{\star} x_k$ stabilizes the system \eqref{eqn:ltisys} and minimizes the infinite horizon LQR cost \eqref{eqn:lqr}.  The optimal cost is given by $J^\star = \mbox{trace}(P\Sigma_w$).
\end{lemma}
\vspace{0.1in}

We use the shorthand notation $\text{LQR}(A,B,\Sigma_w,C,R)$ to refer to the infinite horizon LQR problem with dynamics $(A,B)$, noise covariance $\Sigma_w$, and cost matrices $(C^\top C, R)$. 

We want to quantify the effect of perturbations to the LQR problem data on the solution of the DARE.  Specifically, we consider two LQR problems  $\LQR(A_i,B_i,\Sigma_w,C_i,R)$ with corresponding DARE solutions $P_i$ for $i=1$, $2$.  We require both problems satisfy Assumption~\ref{asmp:lqrReg}. Let $\epsilon>0$ such that:
\begin{align*}
     \| A_1 - A_2 \|_2 \le \epsilon, \,\,\,
    \| B_1 - B_2 \|_2 \le \epsilon, 
    \mbox{ and } \| C_1 - C_2 \|_2 \le \epsilon.
\end{align*}
It follows from \cite[Theorem 2.1]{konstantinov1993perturbation} that the DARE solutions are continuous functions of the problem data. Hence $P_1$ and $P_2$ are ``close" if $\epsilon$ is sufficiently small. Our goal is to quantify this statement. Specifically, we derive a constant $c_1$ such that the DARE solutions satisfy 
\begin{align}
        \| P_1 - P_2 \|_2 \le c_1 \epsilon.     
\end{align}
for sufficiently small $\epsilon$. We provide an explicit expression for $c_1$ in terms of data only from the first LQR problem.

Before proceeding, we comment on our starting assumptions. First, note that both problems share the same process noise $\Sigma_w$. The process noise only affects the optimal LQR cost, not the DARE solutions $P_i$. It does not enter the DARE and hence it does not affect our perturbation analysis. Second, both problems also share the same control cost matrix $R$. This assumption is for simplicity; our analysis could be extended to include perturbations to this matrix. Third, we assume perturbations to the state cost are characterized by the difference in square root factors $C_1-C_2$.  This type of perturbation bound appears naturally in certain problems. Note that $(C,A)$ is detectable as in Assumption~\ref{asmp:lqrReg}(d) if and only if $(Q,A)$ is detectable \cite[Problem 22.5]{hespanha23}.

Finally, we require both problems satisfy Assumptions~\ref{asmp:lqrReg}(c) and (d) throughout this letter. This can be weakened to only assume Problem 1 satisfies \ref{asmp:lqrReg}(c) and (d) when $\epsilon$ is sufficiently small. If $(A_1,B_1)$ is stabilizable then there exists a gain $K_1$ such that $A_1-B_1K_1$ is stable, and therefore all eigenvalues are strictly inside the unit disk.  
Eigenvalues are continuous functions of the matrix entries and so there exists a range of $\epsilon$ for which all eigenvalues of $A_2-B_2K_1$ remain inside the unit disk. For all $\epsilon$ in this range, then $(A_2,B_2)$ is stabilizable. An explicit bound for this range can be derived, but is omitted here. \footnote{Consider the  two closed-loop matrices $\Ahat_{P1} :=A_1-B_1K_1$ and $\Ahat_{P2} :=A_2-B_2K_1$.  The difference in closed-loop matrices is bounded by:
\begin{align*}
    \|\Ahat_{P1}-\Ahat_{P2}\|_2 
    \le \| A_1-A_2\|_2 + \|B_1-B_2\|_2 \|K_1\|_2
    \le \epsilon (1+\|K_1\|_2).
\end{align*}
Next, define the constant  $\chat_{A1}:=\max_{|z|\ge 1} \| (zI-\Ahat_{P1})^{-1}\|_2$. This constant is finite because $\Ahat_{P1}$ is Schur. Define $\epsilon_0:=0.5/(\chat_{A1} (1+\|K_1\|_2))>0$. If $\epsilon < \epsilon_0$ then $\Ahat_{P2}$ is Schur
(Lemma~\ref{lem:resolvent_identities} in the Appendix). Hence $\epsilon < \epsilon_0$ implies that $(A_2,B_2)$ is stabilizable.} 
Similarly, detectability of $(C_1,A_1)$ implies detectability of $(C_2, A_2)$ if $\epsilon$ is sufficiently small.

\section{Perturbation Bounds on the \htwonorm}\label{sec:h2normbounds}
We next introduce the relationship between the LQR problem and a related systems norm. First, consider the following discrete-time LTI system $\Ghat$ defined as:
\begin{align}
\label{eq:Gwtoz}
\begin{split}
    x_{k+1} & = \Ahat x_k + \Bhat w_k \\
    e_k & = \Chat x_k + \Dhat w_k
\end{split}
\end{align}
The corresponding transfer function $\Ghat:\C \to \C^{n_e \times n_w}$ maps the complex number $z\in \C$ to the frequency response:
\begin{align}
    \Ghat(z) := \Chat(zI-\Ahat)^{-1}\Bhat + \Dhat.
\end{align}
The Hardy space \htwo is the Hilbert space of systems that are square-integrable on the unit circle and stable, i.e. all poles inside the disk (Section 3.3 of \cite{dullerud13}). The inner product of two stable LTI systems $\Ghat_1$ and $\Ghat_2$ is defined as:
\begin{align}\label{eqn:H2innerprod}
    \langle \Ghat_1 , \Ghat_2 \rangle := 
\frac{1}{2\pi} \int_0^{2\pi} \tr( \Ghat_1^\star(e^{j\theta}) \Ghat_2(e^{j\theta} )) \, \mathrm{d} \theta.
\end{align}
The \htwonorm generated by this inner product is:
\begin{align}
\|\Ghat\|_2:=\sqrt{ \langle \Ghat,\Ghat \rangle}    
= &\left[  \frac{1}{2\pi} \int_0^{2\pi} \| \Ghat(e^{j\theta}) \|_F^2 \,\mathrm{d} \theta\right]^{1/2}.
\end{align}

We can use the \htwonorm to bound the cost of suboptimal LQR gains. Consider the output of a stable system \eqref{eq:Gwtoz} driven by an input $w_k$ that is i.i.d. zero-mean Gaussian noise with covariance $\Sigma_w$. The steady-state mean-square output is:
\begin{align} \label{eqn:ssCost}
\Jhat := & \lim_{N\to \infty} 
\mathbb{E} \left[ \frac{1}{N} \sum_{k=0}^N e_k^\top e_k \right].
\end{align}

The next lemma connects the \htwonorm and this steady-state mean-square output. This result is known, e.g. see \cite[Section 6.1]{dullerud13} for the continuous-time case.
\begin{lemma}
\label{lem:H2normCost}
Assume the system $\Ghat$ in \eqref{eq:Gwtoz} is stable, i.e. all poles are in the unit disk. Then $\Jhat = \|\Ghat\|_2^2$.
\end{lemma}

This equivalence of the \htwonorm  and the steady-state mean-squared output of a system can be used to express the stochastic LQR cost of a given state feedback controller as the \htwonorm of the closed-loop system. This is formalized in the next lemma. In the lemma,  $R^{1/2}$ denotes the unique symmetric, positive definite square root of $R\succ 0$.

\vspace{0.1in}
\begin{lemma}\label{cor:LQR_H2_equiv}
Consider $\LQR(A,B, \Sigma_w, C, R)$ 
that satisfies Assumption~\ref{asmp:lqrReg}. Let $K$ be any stabilizing (potentially suboptimal) static state feedback controller with stochastic LQR cost $J(K)$. Define the closed-loop system $\Ghat$ by
    \begin{equation} \label{eqn:ltisys_cl}
        \begin{split}
            x_{k+1}&=\Ahat x_k+ \Bhat w_k,\\
            e_k &= \Chat x_k.
        \end{split}
    \end{equation}
    where $\Ahat:=A-BK$, $\Bhat:=I$,
    $\Chat := \m{ C^\top, \,\, (-R^{1/2} K)^\top}^\top$, and $\Dhat:=0$.   
 Then $J(K) = \| \Ghat\|_2^2$.
\end{lemma}
\begin{proof}
    The per-step cost with a state-feedback controller $u_k=-Kx_k$ is 
    \begin{equation}
    x_k^\top Q x_k + u_k^\top R u_k = x_k^\top \left( Q+ K^\top R K \right) x_k
    \end{equation}
    This per-step cost can be expressed as $e_k^\top e_k$ where $e_k := \Chat x_k$.
    Given this definition, the stochastic LQR cost in \eqref{eqn:lqr} is the same as the steady-state cost in \eqref{eqn:ssCost} with the system $\Ghat$ from $w_k$ to $e_k$ as given in \eqref{eqn:ltisys_cl}.
    By Lemma \ref{lem:H2normCost}, this implies the stochastic LQR cost with gain $K$ is equal to the square of the \htwonorm of this closed-loop system $\Ghat$, i.e. $J(K) = \|\Ghat\|_2^2$.
\end{proof}
\vspace{0.1in}

Next, we show that the \htwonorm of two systems are ``close" if their state matrices are ``close".

\vspace{0.1in}
\begin{lemma}
\label{lem:CloseH2}
    Let $\Ghat_1$ and $\Ghat_2$ be two stable, $n_e\times n_w$ LTI systems defined by $(\Ahat_i, \Bhat, \Chat_i, \Dhat)$ for $i=1,2$. Assume there exists $\epsilonhat>0$ such that their state matrices satisfy
    \begin{align*}
     \|\Ahat_1-\Ahat_2\|_2 \le \epsilonhat, \,\,\, 
     \|\Chat_1-\Chat_2\|_2 \le \epsilonhat,  
    \end{align*}
    If $\epsilonhat < \min\{ 0.5/\chat_{A1}, 1\}$ where the constant $\chat_{A1}$ is
    \begin{align}\label{eqn:cA1_def}
        \chat_{A1} := \max_{|z| \ge 1} \| (zI-\Ahat_1)^{-1}\|_2,  
    \end{align}
    then the \htwonorms of the two systems satisfy
    \begin{align}
        \left|\|\Ghat_1\|_2^2 - \|\Ghat_2\|_2^2 \right| 
        \le    \left[ 2\|\Ghat_1\|_2 + \epsilonhat \fhat_1  \right] \, \epsilonhat \fhat_1 .
    \end{align}    
    where $r:=\min(n_e,n_w)$ and 
   \begin{align}\label{eqn:fhat1def}
        \fhat_1:=r \chat_{A1} \|\Bhat\|_2  + 2 r \chat_{A1}^2\left( \| \Chat_1\|_2 + 1 \right)  \|\Bhat\|_2.
    \end{align}
\end{lemma}
\begin{proof}
Using the \htwo inner product \eqref{eqn:H2innerprod}, we have:
\begin{align}
    \|\Ghat_1\|_2^2 - \|\Ghat_2\|_2^2
    = \mbox{Real}\left[
    \langle \Ghat_1+\Ghat_2,\Ghat_1-\Ghat_2\rangle \right].
\end{align}    
We apply Cauchy-Schwartz to the difference in \htwo norms:
\begin{align}
\label{eq:H2CauchyS}
   \left| \|\Ghat_1\|_2^2 - \|\Ghat_2\|_2^2 \right|
   \le \| \Ghat_1+\Ghat_2 \|_2 \, \|\Ghat_1 - \Ghat_2\|_2.
\end{align}  
Define $\Delta_{\Ghat}:= \Ghat_2-\Ghat_1$, and note that $\Ghat_1+\Ghat_2$ is equal to $2\Ghat_1 + \Delta_{\Ghat}$. We can use this expression and the triangle inequality to bound the first term on the right side of \eqref{eq:H2CauchyS} as
\begin{align}
\label{eq:H2Bound2}
   \left| \|\Ghat_1\|_2^2 - \|\Ghat_2\|_2^2 \right|
   \le 
   \left[ 2\|\Ghat_1\|_2 + \|\Delta_{\Ghat}\|_2\right] \, \|\Delta_{\Ghat}\|_2.
\end{align}

Next, we upper bound $\|\Delta_{\Ghat}\|_2$ in terms of the state matrices. We start by bounding the \htwonorm as
\begin{align}
\|\Delta_{\Ghat}\|_2
\le \max_{|z|\ge 1} \|\Delta_{\Ghat}(z)\|_F.
\end{align}
We can then bound the Frobenius norm  of the frequency response matrix by $\|\Delta_{\Ghat}(z)\|_F \le r \cdot\|\Delta_{\Ghat}(z)||_2$ to get
\begin{align}
\|\Delta_{\Ghat}\|_2
\le r \cdot \max_{|z|\ge 1} \|\Delta_{\Ghat}(z)\|_2.
\end{align}
Here, the \htwonorm of a system is on the left while the induced $2$-norm of a matrix is on the right. 
Both systems are stable and so the maximum over $|z| \ge 1$ is finite. The proof is completed by showing $\|\Delta_{\Ghat}(z)\|_2$ is upper bounded by $\epsilonhat \fhat_1 $ for all $|z| \ge 1$.

To simplify notation, define the resolvents $\Rhat_i(z):=(zI-\Ahat_i)^{-1}$ for $i=1,2$. The difference in transfer functions evaluated at any $z\in \C$ is
\begin{align*}
\Delta_{\Ghat}(z) = 
(\Chat_1 -\Chat_2)\Rhat_1(z)\Bhat 
+\Chat_2(\Rhat_1(z)-\Rhat_2(z))\Bhat.
\end{align*}
The triangle inequality and submultiplicativity then yields 
\begin{align}
\label{eq:Gzdiff}
\| \Delta_{\Ghat}(z) \|_2 
\le   \left( \epsilon \chat_A  + \| \Chat_2\|_2 \|\Rhat_1(z)-\Rhat_2(z)\|_2\right) \|\Bhat\|_2
\end{align}
at all $|z| \ge 1$. Recall that we assume the state matrices satisfy $\|\Ahat_1-\Ahat_2\|_2 \le \epsilonhat$ with  $\epsilonhat < 0.5/\chat_{A1}$. By Lemma~\ref{lem:resolvent_identities}, we can bound the resolvent difference as
\begin{align}
\label{eq:Rdiff}
\| \Rhat_1(z)-\Rhat_2(z)\|_2
\le 2\chat_{A1}^2 \epsilonhat.
\end{align}
We can also bound the output matrix of system 2 as:
\begin{align}
\label{eq:Cdiff}
    \| \Chat_2 \|_2 \le \| \Chat_1 \|_2 +\epsilonhat.
\end{align}

Substitute \eqref{eq:Rdiff} and \eqref{eq:Cdiff} into \eqref{eq:Gzdiff} to obtain the following bound on the matrix 2-norm at all $|z| \ge 1$:
\begin{align*}
\| \Delta_{\Ghat} \|_2 
& \le   \epsilonhat \chat_{A1} \|\Bhat\|_2 
  + 2 \epsilonhat  \chat_{A1}^2\left( \| \Chat_1\|_2 + \epsilonhat \right)  \|\Bhat\|_2 .
\end{align*}
By $\epsilonhat<1$, the right side simplifies to
$\| \Delta_{\Ghat}(z) \|_2 \le \epsilonhat \fhat_1$. Apply this in \eqref{eq:H2Bound2} to obtain the 
bound in the lemma statement.
\end{proof}

\section{Main results}
\label{sec:MainResults}

In this section, we apply the connection between the LQR problem and the \htwonorm of a system in order to bound the difference in DARE solutions. We first need a continuity result. 
\vspace{0.1in}
\begin{lemma}\label{lem:chatA2bound}
    Consider $\LQR(A_i,B_i,\Sigma_w,C_i,R)$ for $i=1$, $2$.
    Assume Problems 1  and 2 satisfy Assumption~\ref{asmp:lqrReg}. Let $(P_i,K_i)$ be the corresponding solutions and gains for $i=1,2$. 
    Let $\epsilon>0$ such that:
    \begin{align*}
         \| A_1 - A_2 \|_2 \le \epsilon, \,\,\,
        \| B_1 - B_2 \|_2 \le \epsilon, 
        \mbox{ and } \| C_1 - C_2 \|_2 \le \epsilon.
    \end{align*}
    Define the constants
    \begin{align}\label{eqn:cAi_def}
        \chat_{Ai} := \max_{|z| \ge 1} \| (zI-(A_i-B_i K_i)^{-1}\|_2 \text { for } i=1,2.
    \end{align}
    Then there exists an $\epsilon_2$ such that $\chat_{A2} \leq 2 \chat_{A1}$ for all $\epsilon <\epsilon_2$. 
\end{lemma}
\begin{proof}
    Define $\Ahat_i = A_i - B_i K_i$ for $i=1,2$. The matrix difference can be expressed as
    \[\Ahat_1 -\Ahat_2 = (A_1-A_2) - B_1 K_1 +B_2K_1 - B_2K_1 + B_2K_2. \]
    We can apply the triangle inequality and submultiplicativity to bound the norm of the difference as 
    \begin{align*}
        \|\Ahat_1 -\Ahat_2\|_2 &= \|A_1-A_2\|_2  + \|B_2 - B_1\|_2 \|K_1\|_2 \\
        & \quad + \|B_2\|_2 \|K_2 -K_1\|_2  \\
        &\leq \epsilon + \epsilon \|K_1\|_2 + (\|B_1\|_2 + \epsilon ) \|K_1 -K_2\|_2.
    \end{align*}
    By Lemma \ref{lem:Kbnd}, we can define constants $c_1$ and $c_2$ as in \eqref{eqn:Kdiff_constants} which only depend on Problem 1 data such that
    \begin{equation}
        \|K_1-K_2\|_2 \leq  c_1 \epsilon + c_2 \|P_1-P_2\|_2.
    \end{equation}
    Substituting this into the previous bound, we have  
    \begin{align}
        \|\Ahat_1 -\Ahat_2\|_2 &\leq \left( 1 + \|K_1\|_2 + \|B_1\|_2 + c_1\epsilon \right) \epsilon\nonumber\\
        \label{eqn:AhatDiff}
        &\quad + \left( c_2 \epsilon + c_2 \|B_1\|_2  \right) \|P_1-P_2\|_2
    \end{align}
    By \cite[Theorem 2.1]{konstantinov1993perturbation}, the DARE solutions are continuous functions of the LQR problem data, and so $\|P_1-P_2\|_2 \rightarrow 0$ as $\epsilon \rightarrow 0$. This implies the bound on the right hand side of \eqref{eqn:AhatDiff} approaches 0 as $\epsilon \rightarrow 0$. 
    Hence there exists an $\epsilon_2$ such that $\|\Ahat_1-\Ahat_2\|_2\leq 0.5/\chat_{A1}$ for all $\epsilon<\epsilon_2$. 
    By Lemma \ref{lem:resolvent_identities}, we have that $\chat_{A2}\leq 2\chat_{A1}$ for all $\epsilon < \epsilon_2$. 
\end{proof}
\vspace{0.1in}

We can now focus on the difference $P_2-P_1$. This matrix is real and symmetric, and hence its eigenvalues are real (although not necessarily positive). The next lemma bounds the maximum eigenvalue $\lambda_{\max}(P_2-P_1)$.

\vspace{0.1in}
\begin{lemma}\label{lem:LambdaMax}
    Consider $\LQR(A_i,B_i,\Sigma_w,C_i,R)$ for $i=1$, $2$. Assume Problems 1  and 2 satisfy Assumption~\ref{asmp:lqrReg}. Let $(P_i,K_i)$ be the corresponding solutions and gains for $i=1,2$. 
    Let $\epsilon>0$ such that:
    \begin{align*}
         \| A_1 - A_2 \|_2 \le \epsilon, \,\,\,
        \| B_1 - B_2 \|_2 \le \epsilon, 
        \mbox{ and } \| C_1 - C_2 \|_2 \le \epsilon.
    \end{align*}
    If $\epsilon < \min\{ 0.5/(\chat_{A1}(1+\|K_1\|_2)), 1\}$ for the constant $\chat_{A1}$
    \begin{align}\label{eqn:cA1_def2}
        \chat_{A1} := \max_{|z| \ge 1} \| (zI-(A_1-B_1K_1))^{-1}\|_2,  
    \end{align}
    then $P_2-P_1$ satisfies the eigenvalue bound 
    \begin{equation}
        \lambda_{\max}(P_2-P_1) \leq \left[2\sqrt{\|P_1\|_2} + \epsilon  f_1 \right] \epsilon f_1
    \end{equation}
    where 
    \begin{equation}\label{eqn:f2def}
        f_1 := \left(1+\|K_1\|_2\right) \left( \chat_{A1} + 2 \chat_{A1}^2  \left(  \left\|\m{C_1 \\ -R^{1/2} K_1}\right\|_2 +1 \right) \right).
    \end{equation}
\end{lemma}

\begin{proof}
    First, we express the maximum eigenvalue as the stochastic LQR cost for a specific choice of process noise. 
    By definition of an eigenvalue, there exists a real vector $v_p$ such that $\|v_p\|_2=1$ and 
    \begin{align*}
        v_p^\top (P_2-P_1)v_p = \lambda_{\max}(P_2-P_1).
    \end{align*}
    Construct the stochastic LQR problems with system $(A_i,B_i)$, cost matrices $(C_i^\top C_i, R)$ and process noise covariance $\Sigmahat_w = v_pv_p^\top$. Solving this for the optimal solution produces the same $P_i$ as the original problems, as the DARE does not depend on $\Sigma_w$. However, the stochastic LQR costs are $J_{Pi}(K_i)= \tr(P_i\Sigmahat_w)$ as in Lemma~\ref{lem:LQR}. Then we have
    \begin{align}
        \lambda_{\max}(P_2-P_1) &= v_p^\top (P_2-P_1)v_p \nonumber\\
        &= \tr(P_2\Sigmahat_w)- \tr(P_1\Sigmahat_w) \nonumber\\
        &= J_{P2}(K_2) - J_{P1}(K_1). \nonumber
    \end{align}
    Because $K_2$ is the optimal controller for the second LQR problem, we can upper bound the first term by the stochastic LQR cost incurred by applying the suboptimal controller $K_1$ instead. This yields the bound 
    \begin{equation}\label{eqn:maxeig1}
        \lambda_{\max}(P_2-P_1) \leq J_{P2}(K_1) - J_{P1}(K_1).
    \end{equation}

    We now define the two closed-loop systems $\Ghat_i$ for controller $K_1$ as in \eqref{eqn:ltisys_cl} with $\Ahat_{Pi}:=A_i-B_i K_1$, $\Bhat:=v_p$,
    $\Chat_{Pi} := \m{ C_i^\top, \,\, (-R^{1/2} K_1)^\top}^\top$, and $\Dhat:=0$.  
    Define the constant $\chat_{A1}$ as in \eqref{eqn:cA1_def2}. This is the same constant as defined in \eqref{eqn:cA1_def} for the closed-loop system $\Ahat_{P1}$.
    
    Next we bound the difference $\|\Ahat_{P1}-\Ahat_{P2}\|_2$ using the triangle inequality and submultiplicativity as 
    \begin{align*}
        \|\Ahat_{P1}-\Ahat_{P2}\|_2 &\leq \|A_1-A_2\|_2 + \|B_1-B_2\|_2 \|K_1\|_2\\
        &\leq \epsilon (1+\|K_1\|_2) =: \epsilonhat.
    \end{align*}
    Note that $ \epsilonhat < 0.5/ \chat_{A1}$ by the assumption in the lemma statement. Hence by Lemma~\ref{lem:resolvent_identities}, we have that $\Ahat_{P2}$ is stable and so $K_1$ is a stabilizing controller for Problem 2. 
    Applying Lemma~\ref{cor:LQR_H2_equiv} gives that $J_{Pi}(K_1) := \|\Ghat_{Pi}\|_2^2$ for $i=1,2$. 
    We can substitute into \eqref{eqn:maxeig1} to get 
    \begin{align}
        \lambda_{\max}(P_2-P_1)
        &\leq \left| \|\Ghat_{P2}\|_2^2-\|\Ghat_{P1}\|_2^2 \right|.
    \end{align}    
    Applying Lemma~\ref{lem:CloseH2}, again noting $\epsilonhat<0.5/\chat_{A1}$, yields
    \begin{align}\label{eq:eigBound1} 
    \left|\|\Ghat_{P2}\|_2^2 - \|\Ghat_{P1}\|_2^2 \right| 
        \le    \left[ 2\|\Ghat_1\|_2 + \fhat_1  \epsilonhat \right] \, \fhat_1  \epsilonhat.
    \end{align} 
    for $\fhat_1$ as defined in \eqref{eqn:fhat1def}. We can simplify the expression for $\fhat_1$. Specifically, we can bound $\|\Chat_{P1}-\Chat_{P2}\|_2$ as
    \begin{align*}
        \|\Chat_{P1}-\Chat_{P2}\|_2 &\leq \|C_1-C_2\|_2 \leq \epsilon \leq \epsilonhat.
    \end{align*}
    Moreover, we have that $\|\Bhat_p\|_2 = \|v_p\|_2=1$, and $n_w=1$ so $r=\min\{n_e,n_w\}=1$. These facts yield the simplified $\fhat_1$
    \begin{align}\label{eqn:fhat1_simple}
        \fhat_1:=\chat_{A1}  + 2 \chat_{A1}^2\left( \| \Chat_{P1}\|_2 + 1 \right) .
    \end{align}
    Finally, we have that 
    \[\|\Ghat_1\|_2^2 = \tr (P_1 \Sigma_w) = v_p^\top P_1v_p \leq \|P_1\|_2.\]
    Substituting this and $\epsilonhat = \epsilon (1+\|K_1\|_2)$ into \eqref{eq:eigBound1} gives the bound in the lemma statement.  
\end{proof}
\vspace{0.1in}

This bound on $\lambda_{\max}(P_2-P_1)$ only depends on Problem 1 data. In particular, the constant $\chat_{A1}$ is the $\mathcal{H}_\infty$-norm for the closed-loop resolvent of Problem 1, and can be computed by standard numerical algorithms \cite{boyd1989bisection, bruinsma1990fast}. 

We can reverse the roles of Problems 1 and 2 to get a corresponding equivalent bound on $\lambda_{\max}(P_1-P_2)$ that depends only on Problem 2 data. We now use these facts to provide an explicit bound on $\|P_2-P_1\|_2$ if $\epsilon$ is sufficiently small. We present our main result as the following theorem.

\vspace{0.1in}
\begin{theorem}
    Consider $\LQR(A_i,B_i,\Sigma_w,C_i,R)$ for $i=1$, $2$. 
    Assume Problems 1  and 2 satisfy Assumption~\ref{asmp:lqrReg}. Let $(P_i,K_i)$ be the corresponding DARE solutions and gains for $i=1,2$. 
    Let $\epsilon>0$ such that:
    \begin{align*}
         \| A_1 - A_2 \|_2 \le \epsilon, \,\,\,
        \| B_1 - B_2 \|_2 \le \epsilon, 
        \mbox{ and } \| C_1 - C_2 \|_2 \le \epsilon.
    \end{align*}
    Define the constant $\chat_{A1}$ as in \eqref{eqn:cA1_def2} and 
    \begin{equation}\label{eqn:alphaDef}
        \alpha := \|R^{-1}\|_2 \left( \|B_1\|_2 + 1\right) \left( \|P_1\|_2 + 1\right)  \left( \|A_1\|_2+ 1\right).
    \end{equation}
    Then there exists an $\epsilon_0<1$ such that for all $\epsilon<\epsilon_0$ the DARE solutions satisfy the following bound
    \begin{equation}\label{eqn:Pdiff_3}
        \|P_2-P_1\|_2 \leq \left[2\sqrt{\|P_1\|_2+1} + \epsilon f_0 \right] \epsilon f_0 
    \end{equation}
    where 
    \begin{align}
        f_0 &:=  2 \chat_{A1} (1+\alpha)  \left[ 1 +
        8 \chat_{A1} + 4\chat_{A1} ( \|C_1\|_2  + \|R^{1/2}\|_2 \alpha) \right].\nonumber 
    \end{align}    
\end{theorem}
\vspace{0.1in}
\begin{proof}
    Recall that $P_2-P_1$ is a real symmetric matrix. Its singular values are equal to the absolute values of its eigenvalues. 
    Hence, the max singular value is bounded by 
    \begin{equation}\label{eqn:Pdiff_max}
        \|P_2-P_1\|_2 \leq \max\left[ \lambda_{\max}(P_2-P_1), \lambda_{\max}(P_1-P_2)  \right].
    \end{equation}
    We can bound these expressions by applying Lemma \ref{lem:LambdaMax} for sufficiently small $\epsilon$. 
    We first quantify the bound on $\epsilon$. Define 
    \begin{align}\label{eqn:cA2_def}
        \chat_{A2} := \max_{|z| \ge 1} \| (zI-(A_2-B_2K_2))^{-1}\|_2. 
    \end{align}
    By Lemma \ref{lem:chatA2bound}, there exists an $\epsilon_2<1$ such that $\chat_{A2}\leq 2\chat_{A1}$ for all $\epsilon< \epsilon_2$. 
    Define the following constant
    \begin{align}\label{eqn:eps0_def}
        \epsilon_0:= \min\{0.5/(2\chat_{A1}(1+\alpha)),  \epsilon_2 \}.
    \end{align}
    We will show that the DARE bound holds for all $\epsilon<\epsilon_0$.

    Define $\epsilon_1:=\min\{ 0.5/(\chat_{A1}(1+\|K_1\|_2)), 1\}$. If $\epsilon < \epsilon_1$, then by Lemma~\ref{lem:LambdaMax},  
    \begin{equation}\label{eqn:lambda1}
        \lambda_{\max}(P_2-P_1) \leq \left[2\sqrt{\|P_1\|_2} + \epsilon f_1 \right] \epsilon f_1
    \end{equation}
    where $f_1$ defined in \eqref{eqn:f2def} is 
    \begin{equation*}
        f_1 := \left(1+\|K_1\|_2\right) \left( \chat_{A1} + 2 \chat_{A1}^2  \left(  \left\|\m{C_1 \\ -R^{1/2} K_1}\right\|_2 +1 \right) \right).
    \end{equation*}

    Next, we show $\epsilon_0\leq \epsilon_1$ and hence \eqref{eqn:lambda1} also holds for all $\epsilon<\epsilon_0$. Define $S_1:= R+B_1^\top P_1 B_1$ and note that $S_1 \succeq R \succ 0$ because $P_1\succeq 0$. Hence $S_1^{-1} \prec R^{-1}$, and so $\|S_1^{-1}\|_2 \leq \|R^{-1}\|_2$. Using this fact and submultiplicativity  gives
    \begin{align}
        \|K_1\|_2 &= \left\|( R+B_1^\top P_1 B_1)^{-1} B_2^\top P_2 A_2\right\|_2 \nonumber\\
        &\leq \|R^{-1}\|_2 \|B_1\|_2 \|P_1\|_2 \|A_1\|_2 \leq \alpha \label{eqn:kbound}
    \end{align}
    Therefore, the following inequality also holds:   
    \begin{equation}
        0.5/(2\chat_{A1}(1+\alpha)) \leq  0.5/(\chat_{A1}(1+\|K_1\|_2)).
    \end{equation}
    The left side appears in the definition of $\epsilon_0$ and the right side appears in the definition of $\epsilon_1$, and so $\epsilon_0 \leq \epsilon_1$. 
    Moreover, we have $f_0 \geq f_1$ because $\|K_1\|_2\leq \alpha$, $\chat_{A1}\leq 2\chat_{A1}$, and 
    \begin{align*}
        \left\|\m{C_1 \\ -R^{1/2} K_1}\right\|_2 &\leq \|C_1\|_2 + \|R^{1/2}\|_2 \|K_1\|_2  \\
        &\leq \|C_1\|_2 + 1 + \|R^{1/2}\|_2 \alpha.
    \end{align*}
    
    At this point, we can bound $\|P_2\|_2$ in terms of $\|P_1\|_2$. Both $P_i$ are symmetric positive semidefinite, so the eigenvalues for both matrices are real, nonnegative, and equal to the singular values. This gives $\|P_i\|_2 = \lambda_{\max}(P_i)$ for $i=1,2$. We obtain
    \begin{align}
        \|P_2\|_2 &= \lambda_{\max}(P_2) = \lambda_{\max}(P_2-P_1+P_1) \nonumber \\
        &\leq \lambda_{\max}(P_2-P_1) + \lambda_{\max}(P_1) \nonumber \\
        &\leq \|P_1\|_2 + 1. \label{eqn:Pdiff}
    \end{align}
    Here the first inequality follows from eigenvalue bounds for symmetric matrices \cite[Theorem 8.1.5]{golub2013matrix}. Moreover, by \eqref{eqn:lambda1} the first term can be bounded by 1 if $\epsilon$ is sufficiently small.\footnote{The bound on $\lambda_{\max}(P_2-P_1)$ is linear in $\epsilon$ to first order and all other factors are constants that depend on Problem 1 data. This bound can be assumed less than 1 by enforcing that $\epsilon$ is sufficiently small.}
    
    We can now bound $\lambda_{\max}(P_1-P_2)$. Define $\epsilonhat_2 := 0.5/(\chat_{A2}(1+\|K_2\|_2))$. 
    Then if $\epsilon<\epsilonhat_2$, we can again apply Lemma~\ref{lem:LambdaMax} with the two systems reversed to get
    \begin{equation}\label{eqn:lambda2}
        \lambda_{\max}(P_1-P_2) \leq \left[2\sqrt{\|P_1\|_2+1} + \epsilon f_2 \right] \epsilon f_2
    \end{equation}
    where 
    \begin{equation}\nonumber
        f_2 :=  (1+\|K_2\|_2) \left( 2\chat_{A1}    + 4 \chat_{A1}^2  \left[ \left\|\m{C_2 \\ -R^{1/2} K_2}\right\|_2 +1 \right] \right). 
    \end{equation}
    Note that we have used $\|P_2\|\leq \|P_1\|_2+1$ and $\chat_{A2}\leq 2\chat_{A1}$. 
    
    Next, we show $\epsilon_0 \leq \epsilonhat_2$ and hence \eqref{eqn:lambda2} also holds for all $\epsilon<\epsilon_0$. We bound $\|K_2\|_2$ following the same steps as in \eqref{eqn:kbound}
    \begin{align}
        \|K_2\|_2 \leq \|R^{-1}\|_2 \|B_2\|_2 \|P_2\|_2 \|A_2\|_2 .
    \end{align}
    We apply the triangle inequality and \eqref{eqn:Pdiff} to get 
    \begin{align*}
        \|K_2\|_2 &\leq \|R^{-1}\|_2 (\|B_1\|_2+\epsilon)( \|P_2\|_2+1)( \|A_2\|_2 +\epsilon) \leq \alpha.
    \end{align*}
    Therefore, the following inequality also holds: 
    \begin{equation}
        0.5/(2\chat_{A1}(1+\alpha)) \leq  0.5/(\chat_{A2}(1+\|K_2\|_2)).
    \end{equation}
    Again, the left side is in the definition of $\epsilon_0$ and the right side appears in the definition of $\epsilonhat_2$ and so $\epsilon_0\leq \epsilonhat_2$.  Moreover, we have $f_0 \geq f_2$ because $\|K_2\|_2\leq \alpha$, $\chat_{A2}\leq 2\chat_{A1}$, and 
    \begin{align*}
        \left\|\m{C_2 \\ -R^{1/2} K_2}\right\|_2 &\leq        \|C_2\|_2 + \|R^{1/2}\|_2 \|K_2\|_2  \\
        &\leq \|C_1\|_2 + \epsilon + \|R^{1/2}\|_2 \alpha\\
        &\leq \|C_1\|_2 + 1 + \|R^{1/2}\|_2 \alpha.
    \end{align*}
    
    We have shown bounds on $\lambda_{\max}(P_2-P_1)$ and $\lambda_{\max}(P_1-P_2)$ that hold for all $\epsilon<\epsilon_0$, where $\epsilon_0$ is defined in \eqref{eqn:eps0_def}. Substituting these bounds on the maximum eigenvalues back into \eqref{eqn:Pdiff_max} gives the result in the theorem statement. 
\end{proof}

This bound provides an explicit computable bound on $\|P_2-P_1\|_2$ that holds for sufficiently small $\epsilon$. It is possible to provide an expression for the constant $\epsilon_0$ in terms of Problem 1 data, but it requires a homotopy argument and is omitted for space. 

We compare our resulting bound to \cite[Proposition 1]{mania2019certainty}, which is a variation of the bounds in \cite{konstantinov1993perturbation}. First, define the function
\begin{equation}\label{eqn:mania_tau}
    \tau(M,\rho) := \sup \left\{ \|M^k\|\rho^{-l} \, : \, k\geq 0 \right\}
\end{equation}
that quantifies the growth of powers of a square matrix $M$. We now include their Riccati perturbation bound translated to our notation. 
\begin{proposition}[Prop. 1 from \cite{mania2019certainty}]
     Define the closed-loop matrix $\Ahat_{1}:= A_1-B_1K_1$. Let $\gamma \geq \rho(\Ahat_1)$, and also let $\epsilon$ such that $\|A_2-A_1\|$, $\|B_2-B_1\|$, and $\|Q_2-Q_1\|$ are at most $\epsilon$. Let $\|\cdot \|_+ = \|\cdot\|+1$. We assume that $R\succ 0$, $(A_1,B_1)$ is stabilizable, $(Q_1^{1/2},A_1)$ observable, and $\sigma_{\min}(P_1)\geq 1$. Then 
    \begin{equation}\label{eqn:Mania_Pdiff_bound}
        \|P_2-P_1\|\leq \mathcal{O}(1) \epsilon \frac{\tau(\Ahat_1, \gamma)^2}{1-\gamma^2} \|A_1\|_+^2 \|P_1\|_+^2 \|B_1\|_+^2 \|R^{-1}\|_+, \nonumber
    \end{equation}
    for a universal constant $\mathcal{O}(1)$, as long as    
    \begin{equation}\label{eqn:Mania_eps_bound}\nonumber
    \begin{split}
        \epsilon \leq &\mathcal{O}(1) \frac{(1-\gamma^2)^2}{\tau(\Ahat_1, \gamma)^4} \|A_1\|_+^{-2} \|P_1\|_+^{-2} \|B_1\|_+^{-3} \|R^{-1}\|_+^{-2}  \\
        &\quad\qquad \times \min\{ \|\Ahat_1\|_+^{-2}, \|P_1\|_+^{-1} \}.
    \end{split}\end{equation}
\end{proposition}
\vspace{0.1in}

Both results give a bound on $\|P_2-P_1\|_2$ that is linear in $\epsilon$. We note two main differences. 
First, our bounds are not directly comparable because we employ different constants to account for the stability margin. In particular, their result uses the time-domain metric of closed-loop stability $\tau$, while we use the frequency-domain metric $\chat_{A1}$ as defined in \eqref{eqn:cA1_def2}. This gives different interpretations for the DARE perturbation bound. 
The second major difference is that they assume $(Q_1^{1/2},A_1)$ is observable. Hence $P_1 \succ 0$ and $Q_1$ and $R$ can be scaled without loss of generality to satisfy $\sigma_{\min}(P_1)\ge 1$. Our bound starts from the weaker assumption that $(C_1, A_1)$ is detectable and hence the DARE solution $P_1$ is only guaranteed to be positive semidefinite. As a result, our bound applies in more general situations. 

\section{Conclusion}
In this letter, we derive novel perturbations bounds for the discrete Riccati equation. We leverage the relationship between the time and frequency domain to bound the norm of the Riccati solutions for two systems that are close in the matrix $2$-norm. Future work will include using this approach to derive corresponding bounds for other controls objectives for which suitable frequency domain equivalencies exist.

\useRomanappendicesfalse
\makeatletter
\renewcommand{\@IEEEthmcounterin}[1]{\Alph{#1}}
\makeatother
\appendices
\section{Supplemental Lemmas}
The next lemma bounds the difference between two LQR gains in terms of the difference between Riccati solutions. 

\vspace{0.1in}
\begin{lemma}\label{lem:Kbnd}
    Define the two state feedback gains 
    \begin{equation}
        K_i := (R+B_i^\top P_iB_i)^{-1} B_i^\top P_i A_i, \quad i=1,2
    \end{equation}
    where $R\succ 0 $ and $P_i\succeq 0$. Let $\epsilon>0$ such that
    \begin{align*}
     \|A_1-A_2\|_2 \le \epsilon, \,\,\, 
     \|B_1-B_2\|_2 \le \epsilon.  
    \end{align*}
    If $\epsilon<1$ then the gain difference is bounded by 
    \begin{equation}\label{eqn:Kdiff_bound}
        \|K_1-K_2\|_2 \leq c_1 \epsilon + c_2 \|P_1-P_2\|_2
    \end{equation}
    where 
    \begin{align}\label{eqn:Kdiff_constants}
        c_1 := &\|R^{-1}\|_2  \|K_1\|_2 \|P_1\|_2 \, \left( 2 \|B_1\|_2 + 1 \right) \\
        &\quad + \|R^{-1}\|_2\|P_1\|_2 \, \left( \|A_1\|_2 + \|B_1\|_2 +1 \right) \nonumber\\ 
        c_2 := &\|R^{-1}\|_2 \|K_1\|_2 \, \left( 2 \|B_1\|_2 + 1 \right)^2 \\
        &\quad + \|R^{-1}\|_2 \, \left( \|B_1\|_2 + 1 \right) \, \left( \|A_1\|_2 + 1 \right).\nonumber
    \end{align}
\end{lemma}
\vspace{0.1in}
\begin{proof}
    Define $S_i:=R+B_i^\top P_i B_i$ for $i=1,2$. We note the following relation
    \begin{align*}
        S_2(K_1-K_2) &= (S_2-S_1)K_1 + (S_1 K_1 - S_2 K_2) \\
        &= (S_2-S_1)K_1 + \left(B_1^\top P_1 A_1 - B_2^\top P_2 A_2\right). 
    \end{align*}
    Hence we have that 
    \begin{equation*}
        K_1-K_2 = S_2^{-1} \left[ (S_2-S_1)K_1 + \left(B_1^\top P_1 A_1 - B_2^\top P_2 A_2\right)\right]
    \end{equation*}
    Next note that $S_2 \succeq R \succ 0$ because $P_2\succeq 0$.
    Hence $S_2^{-1} \prec R^{-1}$, and so $\|S_2^{-1}\|_2 \leq \|R^{-1}\|_2$. By submultiplicativity and the triangle inequality, we have 
    \begin{align} \label{eqn:DeltaKBound}
        \|K_1-K_2\|_2 \leq \|R^{-1}\|_2 &\, \left( \|S_2-S_1\|_2 \|K_1\|_2 + \right. \\
        &\left.\, \,  \left \|B_1^\top P_1 A_1 - B_2^\top P_2 A_2\right\|_2 \right) \nonumber
    \end{align}
    We next turn to bounding $\Delta_S:=S_2-S_1$. Define $\deltaB:= B_2 -B_1$ and $\Delta_P:= P_2 -P_1$. Then we have 
    \begin{align*}
        \Delta_S &= B_2^\top P_2 B_2 - B_1^\top P_1 B_1 \\
        &= (B_1+\deltaB)^\top P_2 (B_1+\deltaB) - B_1^\top P_1 B_1 \\
        &= B_1^\top \Delta_P B_1 + \deltaB^\top P_2 B_1 + B_1^\top P_2 \deltaB + \deltaB^\top P_2 \deltaB.
    \end{align*}
     Take the norm of both sides. Using submultiplicativity and the triangle inequality then yields:
    \begin{equation}
    \label{eq:DeltaSBound}
        \|\Delta_S\|_2 \leq \|B_1\|_2^2 \|\Delta_P\|_2 + 2\epsilon \|B_1\|_2 \|P_2\|_2 + \epsilon^2 \|P_2\|_2
    \end{equation}
    Next note that $P_2= P_1+\Delta_P$, so that $\|P_2\|_2 \leq \|P_1\|_2 + \|\Delta_P\|_2$. Substitute this into \eqref{eq:DeltaSBound} to obtain the bound:
    \begin{align*}
        \|\Delta_S\|_2 \leq  \epsilon \|P_1\|_2 \, \left( 2\|B_1\|_2  + \epsilon\right) 
        + \|\Delta_P\|_2  \left(\|B_1\|_2^2 + \epsilon \right)^2
    \end{align*}    
    We can similarly bound $\Delta_{BPA}:= B_1^\top P_1 A_1 - B_2^\top P_2 A_2$. Define $ \deltaB:= B_2 -B_1$ and $ \deltaA:= A_2 - A_1$.  Then we have:
    \begin{align*}
        \Delta_{BPA}    = B_1^\top \Delta_P A_1 + \deltaB^\top P_2 A_1 + B_1^\top P_2 \deltaA + \deltaB^\top P_2 \deltaA 
    \end{align*}
    By submultiplicativity and the triangle inequality, we have 
    \begin{align*}
        \left\|\Delta_{BPA}  \right\|_2 &\leq \|B_1\|_2 \|A_1\|_2 \|\Delta_P\|_2\\
        &\quad + \epsilon \|P_2\|_2 \left(  \|A_1\|_2 + \|B_1\|_2 \right)  +\epsilon^2  \|P_2\|_2  
    \end{align*}
    Again, use $\|P_2\|_2 \leq \|P_1\|_2 + \|\Delta_P\|_2$ to get 
    \begin{align*}
        & \left\|\Delta_{BPA} \right\|_2 \leq \epsilon \|P_1\|_2 \, \left(  \|A_1\|_2 + \|B_1\|_2 +\epsilon \right) \\
        & \hspace{0.4in}
        + \|\Delta_P\|_2 \cdot \left[\|B_1\|_2 \|A_1\|_2    +  \epsilon  \left(  \|A_1\|_2 + \|B_1\|_2 \right) + \epsilon^2 \right].
    \end{align*}
    We can factor this second term to the more concise expression 
    \begin{align}
    \label{eq:DeltaBPABound}
        \left\|\Delta_{BPA} \right\|_2 \leq & \epsilon \|P_1\|_2 \left(  \|A_1\|_2 + \|B_1\|_2 +\epsilon \right) \\
        & 
        \nonumber
        + \|\Delta_P\|_2  \left(\|B_1\|_2 + \epsilon\right) \left(  \|A_1\|_2 + \epsilon \right).
    \end{align}    
    Apply the bounds on $\|\Delta_S\|_2$ and $\|\Delta_{BPA}\|_2$ to the expression in \eqref{eqn:DeltaKBound}. This yields the bound on $\|K_1-K_2\|_2$ given in the lemma statement.
\end{proof}
\vspace{0.1in}
The proof is omitted due to space constraints. 

The next lemma provides resolvent bounds when two matrices $\Ahat_i$ are sufficiently close.

\vspace{0.1in}
\begin{lemma}\label{lem:resolvent_identities}
    Consider a matrix $\Ahat_1$ with all eigenvalues strictly inside the unit disk. Define the resolvent $\Rhat_1(z) := (zI-\Ahat_1)^{-1}$ along with the following constant
    \begin{align}
        \chat_{A1} := \max_{|z| \ge 1} \| (zI-\Ahat_1)^{-1}\|_2.  
    \end{align}        
    Let $\Ahat_2$ be a second matrix such that for some $\epsilon>0$:
    \begin{align}
     \|\Ahat_1-\Ahat_2\|_2 \le \epsilon. 
    \end{align}
    If $\epsilon < 0.5/\chat_{A1}$, then  $\Rhat_2(z) := (zI-A_2)^{-1}$ is well defined for all $|z| \ge 1$ and hence $\Ahat_2$ has all eigenvalues strictly inside the unit disk. Moreover,  $\|\Rhat_2(z) \|_2 \leq 2 \chat_{A1}$ and $\|\Rhat_2(z) - \Rhat_1(z)\|_2 \leq 2 \chat_{A1}^2 \epsilon$ for all $|z| \ge 1$.
\end{lemma}
\begin{proof}
    Consider any $|z| \ge 1$.
    The matrix $(zI-\Ahat_1)$ is invertible for any such $z$ and hence we can write:    \begin{align}\label{eqn:zIAeps}
    \begin{split}
        (zI-\Ahat_2) &= \left[ (zI-\Ahat_1)- (\Ahat_2 -\Ahat_1)\right] \\
        &= (zI-\Ahat_1)\cdot \left[ I-\Rhat_1(z) \cdot \Delta_{\Ahat} \right] 
    \end{split}\end{align}
    where $\Delta_{\Ahat}:= \Ahat_2-\Ahat_1$.
    Next, bound $\Rhat_1(z) \cdot \Delta_{\Ahat}$ as
    \begin{align}\label{eqn:Meps}
        \|\Rhat_1(z) \cdot \Delta_{\Ahat}\|_2 \le \epsilon \|(zI-\Ahat_1)^{-1}\|_2 
        \le \epsilon \chat_{A1} <0.5
    \end{align} 
    where the last inequality follows from the assumption $\epsilon<0.5/\chat_{A1}$. Hence $\|\Rhat_1(z) \cdot \Delta_{\Ahat}\|_2 <1$. Lemma 2.3.3 in \cite{golub2013matrix} then implies $I-\Rhat_1(z) \cdot \Delta_{\Ahat}$ is invertible. 
    % Thus, by \eqref{eqn:zIAeps} the inverse of $(zI-\Ahat_2)$ is well-defined for all  and using
    Using \eqref{eqn:zIAeps} gives
    \begin{align}\label{eqn:R2form}
        \Rhat_2(z) = \left( I-\Rhat_1(z) \cdot \Delta_{\Ahat} \right)^{-1} \Rhat_1(z).
    \end{align}
    Thus, the resolvent $\Rhat_2(z)$ is well-defined for all $|z| \ge 1$ so $\Ahat_2$ has all eigenvalues strictly inside the unit disk. 
    
    Lemma 2.3.3 in \cite{golub2013matrix} also gives the following bound:
    \begin{align*}
        \left\|(I-\Rhat_1(z) \cdot \Delta_{\Ahat})^{-1}\right\|_2 &\le \frac{1}{1-\|\Rhat_1(z) \cdot \Delta\Ahat\|_2} < 2,
    \end{align*}
    where the second inequality uses the upper bound in \eqref{eqn:Meps}. 
    Applying this bound to \eqref{eqn:R2form} gives:
    \begin{align*}
        \|\Rhat_2(z)\|_2 \leq 2\chat_{A1}. 
    \end{align*}
    Finally, the resolvents are related by the following identity:
    \begin{align}
    \label{eq:ResolventDiff}
    \Rhat_1(z)-\Rhat_2(z)
       =\Rhat_1(z) (\Ahat_1-\Ahat_2) \Rhat_2(z).   
    \end{align}
    Applying the bounds $\|\Rhat_1(z)\|_2 \le \chat_{A1}$, $\|\Rhat_2(z)\|_2 \le 2\chat_{A1}$, and $\|\Ahat_1-\Ahat_2\|_2 \le \epsilon$ to \eqref{eq:ResolventDiff} with submultiplicativity yields:
    \begin{align*}
        \|\Rhat_1(z)-\Rhat_2(z)\|_2 
        \leq 2 \chat_{A1}^2 \epsilon. 
    \end{align*}
\end{proof}

\bibliographystyle{IEEEtran}
\bibliography{root}

\end{document}

\typeout{get arXiv to do 4 passes: Label(s) may have changed. Rerun}

%% file: preamble.tex
\usepackage[english]{babel}
\usepackage[utf8]{inputenc}
\usepackage{amsfonts}
\usepackage{xcolor}
\usepackage{comment}
\usepackage{balance}

\usepackage{graphicx}
\usepackage{pifont}
\usepackage[ruled, vlined]{algorithm2e}
\usepackage{xstring}
\usepackage{cite}
\usepackage{multirow}
\usepackage{calc}
\usepackage{adjustbox}
\usepackage{booktabs}

\usepackage{makecell}
\usepackage{bbm}

\let\labelindent\relax
\usepackage{enumitem}
\allowdisplaybreaks
\input{preambleLettersDef.tex}

\makeatletter
\newcommand{\oset}[2]{{\mathpalette\o@set{{#1}{#2}}}}
\newcommand{\o@set}[2]{\o@@set{#1}#2}
\newcommand{\o@@set}[3]{%
  \vbox{\offinterlineskip
    \ialign{\hfil##\hfil\cr
      $\m@th\o@set@demote{#1}#2$\cr
      \noalign{\vskip0.2pt}
      $\m@th#1#3$\cr
    }%
  }%
}
\newcommand{\o@set@demote}[1]{%
  \ifx#1\displaystyle\scriptstyle\else
  \ifx#1\textstyle\scriptstyle\else
  \scriptscriptstyle\fi\fi
}
\makeatother
\newcommand{\vthetaleft}{\oset{\scalebox{0.5}[0.5]{$\bm{\leftarrow}$}}{\vtheta} \kern0em}
\newcommand{\vthetaright}{\oset{\scalebox{0.5}[0.5]{$\bm{\rightarrow}$}}{\vtheta} \kern0em}

\newcommand{\st}{\textup{subject to}}

\newcommand\m[1]{\begin{bmatrix}#1\end{bmatrix}}

\DeclareMathOperator{\LQR}{LQR}

\newcommand{\htwo}{$\mathcal{H}_2$\xspace}
\newcommand{\htwonorm}{$\mathcal{H}_2$-norm\xspace}
\newcommand{\htwonorms}{$\mathcal{H}_2$-norms\xspace}

\makeatletter
\let\save@mathaccent\mathaccent
\newcommand*\if@single[3]{%
  \setbox0\hbox{${\mathaccent"0362{#1}}^H$}%
  \setbox2\hbox{${\mathaccent"0362{\kern0pt#1}}^H$}%
  \ifdim\ht0=\ht2 #3\else #2\fi
  }
\newcommand*\rel@kern[1]{\kern#1\dimexpr\macc@kerna}
\newcommand*\widebar[1]{\@ifnextchar^{{\wide@bar{#1}{0}}}{\wide@bar{#1}{1}}}
\newcommand*\wide@bar[2]{\if@single{#1}{\wide@bar@{#1}{#2}{1}}{\wide@bar@{#1}{#2}{2}}}
\newcommand*\wide@bar@[3]{%
  \begingroup
  \def\mathaccent##1##2{%
    \let\mathaccent\save@mathaccent
    \if#32 \let\macc@nucleus\first@char \fi
    \setbox\z@\hbox{$\macc@style{\macc@nucleus}_{}$}%
    \setbox\tw@\hbox{$\macc@style{\macc@nucleus}{}_{}$}%
    \dimen@\wd\tw@
    \advance\dimen@-\wd\z@
    \divide\dimen@ 3
    \@tempdima\wd\tw@
    \advance\@tempdima-\scriptspace
    \divide\@tempdima 10
    \advance\dimen@-\@tempdima
    \ifdim\dimen@>\z@ \dimen@0pt\fi
    \rel@kern{0.6}\kern-\dimen@
    \if#31
      \overline{\rel@kern{-0.6}\kern\dimen@\macc@nucleus\rel@kern{0.4}\kern\dimen@}%
      \advance\dimen@0.4\dimexpr\macc@kerna
      \let\final@kern#2%
      \ifdim\dimen@<\z@ \let\final@kern1\fi
      \if\final@kern1 \kern-\dimen@\fi
    \else
      \overline{\rel@kern{-0.6}\kern\dimen@#1}%
    \fi
  }%
  \macc@depth\@ne
  \let\math@bgroup\@empty \let\math@egroup\macc@set@skewchar
  \mathsurround\z@ \frozen@everymath{\mathgroup\macc@group\relax}%
  \macc@set@skewchar\relax
  \let\mathaccentV\macc@nested@a
  \if#31
    \macc@nested@a\relax111{#1}%
  \else
    \def\gobble@till@marker##1\endmarker{}%
    \futurelet\first@char\gobble@till@marker#1\endmarker
    \ifcat\noexpand\first@char A\else
      \def\first@char{}%
    \fi
    \macc@nested@a\relax111{\first@char}%
  \fi
  \endgroup
}
\makeatother

\newcommand{\R}{\mathbb{R}}				% Real domain R
\newcommand{\C}{\mathbb{C}}		        % Complex domain C
\newcommand{\dm}[2]
{
	\IfStrEq{#2}{1}{\R^{#1}}{\R^{#1 \x #2}}
}

\newcommand{\tr}{\textup{trace}} 			% Trace
\usepackage{amssymb}
\makeatletter
\newcommand*{\T}{{\mathpalette\@transpose{}}}
\newcommand*{\@transpose}[2]{\raisebox{\depth}{$\m@th#1\intercal$}}
\makeatother

\newcommand*{\x}{\mathsf{x}\mskip1mu} 	

\newcommand{\splitatcommas}[1]{%
	\begingroup
	\begingroup\lccode`~=`, \lowercase{\endgroup
		\edef~{\mathchar\the\mathcode`, \penalty0 \noexpand\hspace{0pt plus .1em}}%
	}\mathcode`,="8000 #1%
	\endgroup
}
\newcommand\Item[1][]{% 
\ifx\relax#1\relax  \item \else \item[#1] \fi
\abovedisplayskip=0pt\abovedisplayshortskip=0pt~\vspace*{-\baselineskip}}

\newcommand*\oline[1]{%
  \vbox{%
    \hrule height 0.5pt%                  % Line above with certain width
    \kern0.25ex%                          % Distance between line and content
    \hbox{%
      \kern-0.1em%                        % Distance between content and left side of box, negative values for lines shorter than content
      \ifmmode#1\else\ensuremath{#1}\fi%  % The content, typeset in dependence of mode
      \kern-0.1em%                        % Distance between content and left side of box, negative values for lines shorter than content
    }% end of hbox
  }% end of vbox
}

\usepackage{hyperref}
\hypersetup{
		colorlinks,
		linkcolor={red},
		citecolor={blue},
		urlcolor={blue}
	}

%% file: preambleLettersDef.tex
\newcommand{\bv}[1]{\mathbf{#1}}		% Bold variable for English letters
\newcommand{\bvgrk}[1]{{\boldsymbol{#1}}}	% Bold variable for Greek letters
\newcommand{\jbar}{\bar{j \phantom{\tiny \,}} \kern -0.1em}

\newcommand{\chat}{\hat{c}}

\newcommand{\fhat}{\hat{f}}

\newcommand{\jhat}{\hat{j \phantom{\tiny \,}} \kern -0.1em}

\newcommand{\jtil}{\tilde{j \phantom{\tiny \,}} \kern -0.1em}

\newcommand{\jchk}{\check{j \phantom{\tiny \,}} \kern -0.1em}

\newcommand{\vj}{\bv{j}}

\newcommand{\vjbar}{\bar{\vj \phantom{\tiny \,}} \kern -0.1em}

\newcommand{\vjhat}{\hat{\vj \phantom{\tiny \,}} \kern -0.1em}

\newcommand{\vjtil}{\tilde{\vj \phantom{\tiny \,}} \kern -0.1em}

\newcommand{\vjchk}{\check{\vj \phantom{\tiny \,}} \kern -0.1em}

\newcommand{\Ahat}{\hat{A}}
\newcommand{\Bhat}{\hat{B}}
\newcommand{\Chat}{\hat{C}}
\newcommand{\Dhat}{\hat{D}}

\newcommand{\Ghat}{\hat{G}}

\newcommand{\Jhat}{\hat{J}}

\newcommand{\Rhat}{\hat{R}}

\newcommand{\Acal}{\mathcal{A}}

\newcommand{\Ycal}{\mathcal{Y}}

\newcommand{\deltaA}{\Delta_A}
\newcommand{\deltaB}{\Delta_B}

\usepackage{bm}

\newcommand{\vAcalhat}{\hat{\bm{\Acal} \phantom{a}} \kern-0.5em}
\newcommand{\vAcalcheck}{\check{\bm{\Acal} \phantom{a}} \kern-0.5em}
\newcommand{\vAcalbar}{\bar{\bm{\Acal} \phantom{a}} \kern-0.5em}
\newcommand{\vYcalhat}{\hat{\bm{\Ycal} \phantom{.}} \kern-0.2em}

\newcommand{\epsilonhat}{\hat{\epsilon}}

\usepackage{upgreek}
\newcommand{\vtheta}{\bvgrk{\uptheta}}

\newcommand{\Sigmahat}{\hat{\Sigma}}